\documentclass[conference]{IEEEtran}
\IEEEoverridecommandlockouts

\usepackage[english]{babel}
\usepackage[T1]{fontenc}
\usepackage{amssymb, amsmath, amsthm}
\usepackage{mathtools}
\usepackage{cite}
\usepackage{graphicx}
\usepackage[caption=false,font=footnotesize,farskip=0pt]{subfig}
\usepackage{xcolor}
\usepackage{multirow}
\usepackage[ruled,vlined,linesnumbered]{algorithm2e}
\usepackage{enumitem}
\usepackage{stfloats}
\usepackage[normalem]{ulem}

\newcommand{\norm}[2]{\| #1 \|_{#2}}
\newcommand{\expval}[1]{\mathbb{E} \left \{ #1 \right\}}

\newcommand{\diag}[1]{\mathrm{diag} \left ( #1 \right)}

\newcommand{\complexset}{\mathbb{C}}
\newcommand{\realset}{\mathbb{R}}
\newcommand{\integerset}{\mathbb{Z}}
\newcommand{\set}[1]{\left\{ #1 \right\}}
\newcommand{\realpart}[1]{\mathfrak{R}\left\{ #1 \right\}}

\newcommand{\uniformdist}[2]{\sim \mathcal{U} \left( #1, #2 \right)}
\newcommand{\sinc}{\mathrm{sinc}}

\newcommand{\floor}[1]{\left\lfloor #1 \right\rfloor}

\theoremstyle{plain}
\newtheorem{lem}{Lemma}

\theoremstyle{definition}
\newtheorem{defn}{Definition}

\theoremstyle{remark}
\newtheorem{rem}{Remark}

\usepackage[acronym,,nohypertypes={acronym,notation}]{glossaries}
\newacronym{3d}{3D}{three dimensional}
\newacronym[plural=ACKs]{ack}{ACK}{acknowledgment}
\newacronym{aoa}{AoA}{angle of arrival}
\newacronym{awgn}{AWGN}{additive white Gaussian noise}
\newacronym{aod}{AoD}{angle of departure}

\newacronym{rat}{RAT}{reflected-angular training}
\newacronym[plural=APs, firstplural=access points (APs)]{ap}{AP}{access point}
\newacronym{b5g}{B5G}{Beyond-5G}
\newacronym[plural=BSs, firstplural=base stations (BSs)]{bs}{BS}{base station}
\newacronym{cc}{CC}{control channel}
\newacronym{ce}{CE}{configuration estimation}
\newacronym{csi}{CSI}{channel state information}
\newacronym{cdf}{cdf}{cumulative distribution function}
\newacronym{crc}{CRC}{cyclic redundancy check}
\newacronym{crlb}{CRLB}{Cram\'er-Rao lower bound}
\newacronym{dc}{DC}{direct current}
\newacronym{dsp}{DSP}{digital signal processing}
\newacronym{dl}{DL}{downlink}
\newacronym{doa}{DoA}{direction-of-arrival}

\newacronym{embb}{eMBB}{enhanced mobile broadband}
\newacronym{emf}{EMF}{electromagnetic field}
\newacronym{em}{EM}{electromagnetic}

\newacronym{fp}{FP}{fractional program}
\newacronym{glrt}{GLRT}{generalized likelihood ratio test}
\newacronym[plural=HRISs, firstplural=Hybrid Reconfigurable Intelligent Surfaces (HRISs)]{hris}{HRIS}{hybrid reconfigurable intelligent surface}
\newacronym{iid}{i.i.d.}{independent and identically distributed}
\newacronym{ios}{IoS}{Internet-of-Surfaces}
\newacronym{iot}{IoT}{Internet-of-Things}
\newacronym[plural=KPIs, firstplural=key performance indicators (KPIs)]{kpi}{KPI}{key performance indicator}

\newacronym{ls}{LS}{least-squares}
\newacronym{lf}{LF}{low frequency}
\newacronym{los}{LoS}{line-of-sight}
\newacronym{lti}{LTI}{linear time invariant}

\newacronym{mac}{MAC}{medium access control}
\newacronym{mimo}{MIMO}{multiple-input multiple-output}
\newacronym{mmimo}{M-MIMO}{massive MIMO}
\newacronym{miso}{MISO}{multiple-input single-output}
\newacronym{ml}{ML}{machine learning}
\newacronym{mle}{ML}{maximum-likelihood estimator}
\newacronym{mmse}{MMSE}{minimum mean squared error}
\newacronym{mmtc}{mMTC}{massive machine-type communications}
\newacronym{mrc}{MRC}{maximum-ratio combining}
\newacronym{mse}{MSE}{mean-squared error}
\newacronym{nlos}{NLoS}{non-line-of-sight}

\newacronym{ofdm}{OFDM}{orthogonal frequency-division multiplexing}

\newacronym{pla}{PLA}{planar linear array}
\newacronym{pap}{P\&P}{plug-and-play}
\newacronym{ppp}{PPP}{Poisson point process}

\newacronym{ra}{RA}{random access}
\newacronym{rap}{RAP}{random access procedure}
\newacronym{rf}{RF}{radio frequency}
\newacronym{rmse}{RMSE}{root-mean-square error}
\newacronym{rss}{RSS}{received signal strength}

\newacronym{se}{SE}{squared error}

\newacronym{sdp}{SDP}{semidefinite programming}
\newacronym{sdr}{SDR}{semidefinite relaxation}
\newacronym{sic}{SIC}{successive interference cancellation}
\newacronym{sinr}{SINR}{signal-to-interference-plus-noise ratio}
\newacronym{smse}{SMSE}{sum mean squared error}
\newacronym{sdma}{SDMA}{space-division multiple-access}
\newacronym{snr}{SNR}{signal-to-noise ratio}
\newacronym{soa}{SoA}{state-of-the-art}
\newacronym{sre}{SRE}{smart radio environment}

\newacronym{toa}{ToA}{time-of-arrival}
\newacronym[plural=UEs, firstplural=users' equipment (UEs)]{ue}{UE}{user's equipment}
\newacronym{tdm}{TDM}{time-division multiplexing}
\newacronym{tdma}{TDMA}{time-division multiple access}
\newacronym{tdd}{TDD}{time-division duplex}
\newacronym{tem}{TEM}{transverse electromagnetic mode}

\newacronym{uatf}{UatF}{use-and-then-forget}
\newacronym{ul}{UL}{uplink}
\newacronym{ula}{ULA}{uniform linear array}
\newacronym{upa}{UPA}{uniform planar array}
\newacronym{urllc}{URLLC}{ultra-reliable-low-latency communication}

\newacronym{mr}{MR}{Maximal-ratio}

\newacronym{acf}{ACF}{autocorrelation function}
\newacronym{ccf}{CCF}{cross-correlation function}
\newacronym{rhs}{RHS}{right-hand side}
\newacronym{lhs}{LHS}{left-hand side}
\newacronym{ar1}{AR$(1)$}{first-order autoregressive}
\newacronym{wss}{WSS}{wide-sense stationary}
\newacronym[plural=pdfs]{pdf}{pdf}{probability density function}
\newacronym{tx}{Tx}{transmitter}
\newacronym{rx}{Rx}{receiver}
\newacronym[plural=RISs]{ris}{RIS}{reconfigurable intelligent surface}

\fnbelowfloat

\def\BibTeX{{\rm B\kern-.05em{\sc i\kern-.025em b}\kern-.08em
    T\kern-.1667em\lower.7ex\hbox{E}\kern-.125emX}}

\begin{document}

\title{Randomized Control of Wireless Temporal Coherence via Reconfigurable Intelligent Surface
    \thanks{
        J. H. I. de Souza and T. Abrão are with the Department of Electrical Engineering, Universidade Estadual de Londrina, Londrina, Brazil; E-mail: joaohis@outlook.com and taufik@uel.br.
            
        V. Croisfelt, F. Saggese, and P. Popovski are with the Department of Electronic Systems, Aalborg University, Aalborg, Denmark; E-mail: \{vcr, fasa, petarp\}@es.aau.dk.
    }
}

\author{
    \IEEEauthorblockN{ 
        João Henrique Inacio de Souza,
        Victor Croisfelt,
        Fabio Saggese,
        Taufik Abrão,
        and Petar Popovski
    }
}

\maketitle

\begin{abstract}
    A reconfigurable intelligent surface (RIS) can shape the wireless propagation channel by inducing controlled phase shift variations to the impinging signals. Multiple works have considered the use of \acrshort{ris} by time-varying configurations of reflection coefficients. In this work we use the \acrshort{ris} to control the channel coherence time and introduce a generalized discrete-time-varying channel model for \acrshort{ris}-aided systems. We characterize the temporal variation of channel correlation by assuming that a configuration of \acrshort{ris}' elements changes at every time step. The analysis converges to a randomized framework to control the channel coherence time by setting the number of \acrshort{ris}' elements and their phase shifts. The main result is a framework for a flexible block-fading model, where the number of samples within a coherence block can be dynamically adapted.
\end{abstract}

\begin{IEEEkeywords}
    Reconfigurable intelligent surface (RIS), temporal correlation, fading.
\end{IEEEkeywords}

\section{Introduction}\label{sec:introduction}

A \gls{ris} consists of a massive number of passive reflecting elements able to alter the phase shifts and possibly the amplitude of impinging wireless signals~\cite{bjornson2022}, thereby exerting control over the wireless propagation. Some RIS instances can be seen as passive holographic \gls{mimo} surfaces. Numerous use cases have been proposed to show how such control can benefit the communication between a \gls{tx} and a \gls{rx}, where the prevailing focus is on improving the communication performance~\cite{bjornson2022}. However, relatively few works explore how to use the \gls{ris} to induce changes that induce temporal diversity in the wireless channel and avoid prolonged unfavorable propagation to a given user.  

To illustrate, consider the channel aging problem that occurs due to \gls{rx} mobility which makes the \gls{csi} acquired by the \gls{tx} unreliable over time. Works as~\cite{chen2022,zhang2022,jiang2023,papazafeiropoulos2022} suggest the use of an \gls{ris} to deal with this problem by compensating for the Doppler effects of mobility. The focus of these works typically relies on optimizing the \gls{ris}' configurations of the elements aiming to minimize the channel aging effect. In~\cite{matthiesen2021}, a continuous-time propagation model is given and is used to configure the \gls{ris} in such a way that the received power is maximized whereas the delay and Doppler spread are minimized. 
The authors in~\cite{sun2021} study the spatial-temporal correlation implied by the \gls{ris} when it is embedded in an isotropic scattering environment. Nevertheless, these prior works do not analyze how the temporal channel statistics, such as the coherence time, can be modeled as a function of the properties of the \gls{ris}' elements. In this paper, we focus on studying how the \gls{ris} can shape temporal channel statistics by relying on a discrete-time-varying channel model. A closely related work is~\cite{Besser2021} where the authors proposed an \gls{ris} phase hopping scheme with the purpose of transforming a slow-fading into a fast-fading channel. This was achieved by randomly varying the \gls{ris}' configurations, significantly improving the outage performance without the need for \gls{csi} at the \gls{ris}. Nevertheless, this paper focuses on how the random variation of phases impacts the outage performance, while here we analyze temporal channel correlation that stems from the properties of the \gls{ris}.

We propose a generalized discrete-time{-varying} channel model for \gls{ris}-aided communication systems, showing {how the} part of the propagation environment controlled by the \gls{ris} shapes the discrete temporal channel statistics. We characterize the temporal variation of channel correlation as the \gls{ris}' {reflections} configuration changes {at every} time index. This analysis reveals how one can control the coherence time of the channel by changing the number of \gls{ris}' reflecting elements and their {phase shift} configurations. Our findings corroborate the results from~\cite{Besser2021}, and prove the possibility of using the \gls{ris} to generate a flexible block-fading model. 

\noindent\textit{Notation}. Boldface lowercase $\mathbf{a}$ and uppercase $\mathbf{A}$ letters represent vectors and matrices, respectively. Calligraphic letters $\mathcal{A}$ represent finite sets. Operators: transpose by $\{\cdot\}^T$, complex conjugate by $\{\cdot\}^*$, and real part by $\mathfrak{R} \{\cdot\}$. Important functions are: $\floor{\cdot}$ the floor function, $\delta[\cdot]$ the Kronecker's delta function, and $\sinc(\theta)= \frac{\sin(\theta)}{\theta}$. The expected value operator is $\expval{\cdot}$ and, unless otherwise stated, it is taken w.r.t. the variable $k$. The complex Gaussian distribution is denoted as $\mathcal{CN} \left( \mu, \sigma^2 \right)$ with mean $\mu$ and variance $\sigma^2$, whereas a uniform random distribution over the range $[a,b]$ is $\mathcal{U} \left( a, b \right)$.

\section{System Model}\label{sec:system-model}

The communication setup consists of one single-antenna \gls{tx}, one single-antenna \gls{rx}, and one \gls{ris} with $N \in \integerset_+$ passive reflecting elements, operating in narrowband communication channel and in free space. The wireless channel consists of two distinct radio paths from the \gls{tx} to the \gls{rx}, the direct path and the reflected path controlled  by the \gls{ris}, see Fig.~\ref{fig:wireless-environment}. The index of the complex samples in the discrete-time domain is denoted by $k \in \integerset$.

Considering the downlink, let $h_{\mathrm{D}}[k] \in \complexset$ denote the channel coefficient from the \gls{tx} to the \gls{rx}, $g_n[k] \in \complexset$ denote the channel coefficient from the \gls{tx} to the $n$-th reflecting element of the \gls{ris}, and $h_n[k] \in \complexset$ denote the channel coefficient from the $n$-th reflecting element of the \gls{ris} to the \gls{rx}, $n \in \set{1,\dots,N}$. Let the channel vectors from the \gls{tx} to the \gls{ris} and from the \gls{ris} to the \gls{rx} be $\mathbf{g}[k] = [g_1[k],\dots, g_N[k]]^T$ and $\mathbf{h}[k] = [ h_1[k],  \dots , h_N[k] ]^T$, respectively. Consider then that the $n$-th reflecting element of the \gls{ris} can induce a phase shift of $\phi_n[k] \in [0,2\pi)$ upon the impinging signal with marginal impact on the signal's amplitude. Thus, we denote the \gls{ris}' configuration impressed at time $k$ as the matrix $\boldsymbol{\Psi}[k] = \diag{[\psi_1[k],\dots,\psi_N[k]]^T}$, whose $n$-th diagonal entry is the \textit{reflection coefficient} $\psi_n[k] = e^{-j\phi_n[k]}$ of the $n$-th \gls{ris} element.\footnote{The \gls{ris} is considered to have unitary attenuation to simplify the presentation. The generalization is considered to be straightforward.} Using the narrowband model from~\cite{bjornson2022}, we assume that a \gls{ris} configuration is constant within the time of each sample $k$.\footnote{During the switching time from a configuration to another, the \gls{ris} can generate an unpredictable channel behavior. Throughout the paper, we assume that this effect is negligible considering that the switching time is much lower than the time reserved for each configuration.} Moreover, we assume that the \gls{ris}' elements have a flat frequency response, preserving the coherence bandwidth of the equivalent channel, $B_c > 0$.\footnote{$B_c$ is inversely proportional to the \textit{channel delay spread}, determined by the multiple reflection delays of the signal between the \gls{tx} and \gls{rx}.} The equivalent channel $h_{\mathrm{eq}}[k] \in \complexset$ from the \gls{tx} to the \gls{rx} is then~\cite{bjornson2022}:
\begin{align}
    \nonumber
    h_{\mathrm{eq}}[k] &= h_{\mathrm{D}}[k] + \mathbf{g}^T[k] \boldsymbol{\Psi}[k] \mathbf{h}[k],\\
    \label{eq:equivalent-channel}
    &= \underbrace{\vphantom{\sum_{n=1}^N}h_{\mathrm{D}}[k]}_{\text{Uncontrollable}} + \underbrace{\sum_{n=1}^N g_n[k] h_n[k] \psi_n[k]}_{\text{Controllable}}.
\end{align}
Eq. \eqref{eq:equivalent-channel} is made up of two different time-variant terms: the \textit{uncontrollable component} depending only on the properties of the propagation environment, and the \textit{controllable components} controlled by changing the \gls{ris}' configurations. 

\begin{figure}
    \centering
    \includegraphics[width=\columnwidth]{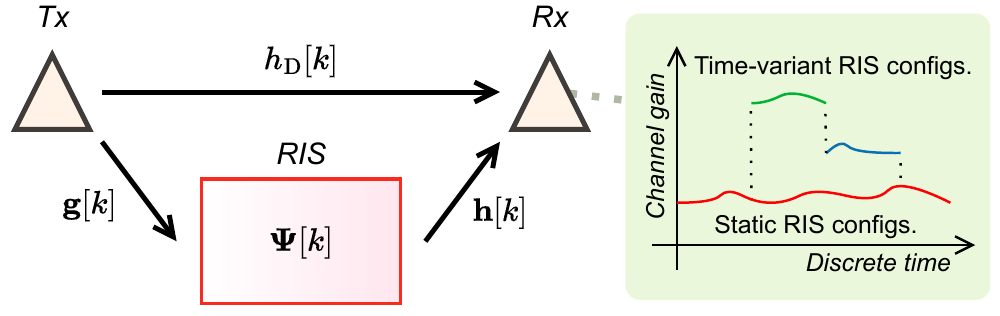}
    \vspace{-7mm}
    \caption{RIS-aided communication system, where the \gls{ris}' elements imposing time-variant reflection configurations can alter the channel {response}.}
    \label{fig:wireless-environment}
\end{figure}

\section{Discrete-Time-Varying Channel Model}

In this section, we present a discrete-time-varying model for the channels $h_{\mathrm{D}}[k]$, $\mathbf{g}[k]$, and $\mathbf{h}[k]$ in order to understand the impact of switching the reflection coefficients of the \gls{ris} over the time index $k$. Based on such a model, we show how the temporal correlation of the equivalent channel behaves by characterizing its \gls{acf} w.r.t. $k$.

\subsection{The Model}

Without loss of generality, let $h[k] \in \complexset$ denote a generic channel coefficient sample from one of the channels $h_{\mathrm{D}}[k]$, ${g}_n[k]$, and ${h}_n[k]$. To represent both \gls{los} and \gls{nlos} components, we assume the time-variant Rician fading model as follows \cite{kurt2022}:
\begin{equation}
    \label{eq:rician-model}
    h[k] = \bar{h} + \check{h}[k],
\end{equation}
where $\bar{h} \in \complexset$ denotes the time-invariant \gls{los} channel component and $\check{h}[k] \in \complexset$ denotes the time-variant \gls{nlos} channel component. Specifically, the \gls{nlos} component is modeled as a stationary \gls{ar1} random process with recurrence relation \cite{baddour2005,truong2013,Wang1996}\footnote{
In \cite{baddour2005}, the authors demonstrated that the AR model could be considered for the computer simulation of correlated fading channels, corroborating that low orders are appropriate for narrowband Doppler fading processes. Moreover, \cite{truong2013,Wang1996} unveil that an \gls{ar1} model is enough to capture most of the channel tap dynamics.}
\begin{equation}
    \label{eq:multipath-component}
    \check{h}[k] = \alpha_h \check{h}[k-1] + \sqrt{1-\alpha_h^2} \sigma_h w[k],
\end{equation}
where $0 \leq \alpha_h < 1$ denotes the \gls{ar1} parameter, $\sigma_h^2 > 0$ denotes the power of the \gls{nlos} component, and {$w[k]$ is a white stationary complex Gaussian process such that} $w[k] \sim \mathcal{CN} \left(0, 1\right)$. From \eqref{eq:rician-model} and the properties of the \gls{ar1}, it is straightforward to demonstrate that:
\begin{equation}
    \label{eq:rician-model-moments}
    \expval{h[k]} = \bar{h} \text{ and } \expval{\left|h[k]\right|^2} = \left|\bar{h}\right|^2 + \sigma_h^2.
\end{equation}
Hence, the Rice factor $\kappa_h > 0$ of $h[k]$ is defined as the ratio between the powers of the \gls{los} and \gls{nlos} components,
\begin{equation}
    \label{eq:rice-factor}
    \kappa_h = \left|\bar{h}\right|^2/\sigma_h^2.
\end{equation}
Given the channel model and the relationship between the powers of its time-invariant and -variant parts, we now define what are the intrinsic parameters to the environment, meaning that they are determined by the physical properties of the wireless propagation medium and the setup geometry, and \textit{cannot} be directly controlled by the system designer.

\begin{defn}
    (Set of Environmental Parameters) We denote as $\mathcal{E}_h = \set{\alpha_h, \kappa_h, {\sigma_h^2}}$  the set of environmental parameters.
\end{defn}

\subsection{Temporal Correlation of the Equivalent Channel}

We now carry out an analysis of the correlation among the channel samples over the time index $k$. Our first result is summarized in the following lemma.

\begin{figure*}[t]
    \begin{align}
        \label{eq:acf-equivalent-channel}
        R_{h_{\mathrm{eq}}h_{\mathrm{eq}}}[\tau] =
        \expval{h_{\mathrm{eq}}[k]h_{\mathrm{eq}}^*[k-\tau]} &=
        \left|\bar{h}_{\mathrm{D}}\right|^2 + \alpha_{h_{\mathrm{D}}}^{|\tau|} \sigma_{h_{\mathrm{D}}}^2 + \sum_{n = 1}^N \left(\left|\bar{g}_n\right|^2 + \alpha_{g_n}^{|\tau|} \sigma_{g_n}^2\right) \left(\left|\bar{h}_n\right|^2 + \alpha_{h_n}^{|\tau|} \sigma_{h_n}^2\right) R_{\psi_n\psi_n}[\tau] \; +\\
        \nonumber
        &+ \sum_{n = 1}^N \sum_{\substack{n' = 1 \\ n' \neq n}}^N \bar{g}_n\bar{g}_{n'}^* \bar{h}_n\bar{h}_{n'}^* R_{\psi_n\psi_{n'}}[\tau] +
        2 \realpart{\bar{h}_{\mathrm{D}}^* \sum_{n = 1}^N \bar{g}_n \bar{h}_n \expval{\psi_n[k]}}.
    \end{align}
    \hrule
\end{figure*}

\begin{figure*}[b]
    \hrule
    \begin{align}
        \nonumber
        R_{h_{\mathrm{eq}}'h_{\mathrm{eq}}'}[\tau] &= \left( 1 - \sinc^2(\theta) \right) \left( \sum_{n = 1}^N \left| g_n \right|^2 \left| \bar{h}_n \right|^2 + \sigma^2 \norm{\mathbf{\bar{g}}}{2}^2 \right) \delta[\tau] + \sinc^2(\theta) \left( \left| \mathbf{\bar{g}}^T\mathbf{\bar{h}} \right|^2 + \alpha^{|\tau|} \sigma^\norm{\mathbf{\bar{g}}}{2}^2 \right),\\
        \tag{12}
        \label{eq:acf-equivalent-channel-os2}
        &\overset{\mathrm{(a)}}{=} \sigma^2 \norm{\mathbf{\bar{g}}}{2}^2 \left[ \left( 1 - \sinc^2(\theta) \right) (\kappa + 1) \delta[\tau] + \sinc^2(\theta) \left( N\kappa\eta + \alpha^{|\tau|} \right) \right].
    \end{align}
    \begin{equation}
        \tag{14}
        \label{eq:correlation-coefficient-os2}
        \rho[\tau] = \frac{R_{h_{\mathrm{eq}}'h_{\mathrm{eq}}'}[\tau]}{R_{h_{\mathrm{eq}}'h_{\mathrm{eq}}'}[0]} = \frac{\left( 1 - \sinc^2(\theta) \right) (\kappa + 1) \delta[\tau] + \sinc^2(\theta) \left( N\kappa\eta + \alpha^{|\tau|} \right)}{\left( 1 - \sinc^2(\theta) \right) (\kappa + 1) + \sinc^2(\theta) \left( N\kappa\eta + 1 \right)}.
    \end{equation}
\end{figure*}

\begin{lem}
    \label{lem:acf}
    Consider that the channel coefficients $h_{\mathrm{D}}[k]$, $g_n[k]$, and $h_n[k]$, $\forall n$, follow the time-variant Rician model in eq. \eqref{eq:rician-model}. Then, the \gls{acf} $R_{h_{\mathrm{eq}}h_{\mathrm{eq}}} : \integerset \rightarrow \realset$ of the equivalent channel is given by eq. \eqref{eq:acf-equivalent-channel} at the top of the next page, where $R_{\psi_n\psi_{n'}} : \integerset \rightarrow \realset$ is the \gls{ccf} of the \gls{ris}' reflection coefficients, calculated for the discrete-time {delay $\tau$} as:
    \begin{equation}
        \label{eq:ccf-reflection-coefficients}
        R_{\psi_n\psi_{n'}}[\tau] = \expval{\psi_n[k]\psi_{n'}^*[k-\tau]}.
    \end{equation}
\end{lem}

\begin{proof}[Proof]
    See Appendix \ref{apx:acf-equivalent-channel}.
\end{proof}

From the above result, one can note that the channel \gls{acf} inherits the \textit{uncontrollable} part and the \textit{controllable} part from the equivalent channel in eq. \eqref{eq:equivalent-channel}.

\section{A Randomized Framework for Controlling the Temporal Correlation}\label{sec:rand-framework}

We start by providing a general framework that describes how to control the temporal correlation by using Lemma~\ref{lem:acf} to set the number of \gls{ris}' reflecting elements $N$ and/or designing their configuration $\set{\psi_n[k]}_{n = 1}^N$. Note that through $N$ we select an \gls{ris} of sufficient size to meet the system's temporal correlation requirements.
Then, we study the case with uniformly distributed phase shifts, as in~\cite{Besser2021}.

\subsection{Temporal Correlation under Random Phase Shifts}

We first make the following simplifying assumptions. \textbf{(1)} There is no direct path from the \gls{tx} to the \gls{rx}, \textit{i.e.}, $h_{\mathrm{D}}[k] = 0$. This holds when obstacles  block the direct path between the \gls{tx} and the \gls{rx}, as in dense urban scenarios and industries. \textbf{(2)} The \gls{los} components are predominant in the channels from the \gls{tx} to the \gls{ris}, \textit{i.e.}, $\kappa_{g_n} \rightarrow \infty$ {dB} $, \forall n$. Therefore, these channels are static, \textit{i.e.}, $g_n[k] = \bar{g}_n, \forall n$.
This can be justified by the fact that the \gls{tx} and the \gls{ris} do not move and the deployment of the \gls{ris} is chosen so as to enhance the \gls{los} components between the \gls{tx} and the \gls{ris}. Using these assumptions, the equivalent channel in \eqref{eq:equivalent-channel} can be rewritten as:
\begin{equation}
    h_{\mathrm{eq}}'[k] = \mathbf{\bar{g}}^T \boldsymbol{\Psi}[k] \mathbf{h}[k],
\end{equation}
and its \gls{acf} is given by:
\begin{align}
    \label{eq:acf-equivalent-channel-simplified}
    \hspace{-3mm} R_{h_{\mathrm{eq}}'h_{\mathrm{eq}}'}[\tau] &= \sum_{n = 1}^N \left|\bar{g}_n\right|^2 \left( \left|\bar{h}_n\right|^2 + \alpha_{h_n}^{|\tau|} \sigma_{h_n}^2 \right) R_{\psi_n\psi_n}[\tau] \, +  \nonumber \\
    & + \sum_{n = 1}^N \sum_{\substack{n' = 1 \\ n' \neq n}}^N \bar{g}_n\bar{g}_{n'}^* \bar{h}_n\bar{h}_{n'}^* R_{\psi_n\psi_{n'}}[\tau],
\end{align}
where $\mathbf{\bar{g}} = [\bar{g}_1, \dots, \bar{g}_N ]^T$ and $\mathbf{\bar{h}} = [ \bar{h}_1, \dots, \bar{h}_N ]^T$. We further assume that the $N$ channels from the \gls{ris}' elements to the \gls{rx} share the same set of environmental parameters $\mathcal{E}= \set{\alpha, \kappa, \sigma^2}$, where $\alpha_{h_n} = \alpha$, $\kappa_{h_n} = \kappa$, and $\sigma_{h_n}^2 = \sigma^2, \, \forall n$. This assumption is valid when considering that the process which introduces the time variations is the same for all {$N$} channels and that the receptions occur in the far-field regime~\cite{baddour2005}. Now, let us assume that the phase shifts at a given time $k$ are drawn from a uniform random distribution:
\begin{equation}
    \label{eq:ris-phase-shifts}
    \phi_n'[k] \uniformdist{\pi - \theta}{\pi + \theta}, \, \forall n,
\end{equation}
where $\theta\in[0,\pi]$ is the {phase shifts distribution parameter}. Then, the \gls{ccf} of the \gls{ris}' reflection coefficients is:
\begin{equation}
    \label{eq:ccf-reflection-coefficients-os2}
    R_{\psi_n'\psi_{n'}'}[\tau] =
    \begin{cases}
        1, & \textrm{if} \; n = n'\\
        \sinc^2(\theta), & \textrm{otherwise}
    \end{cases},
\end{equation}
where we used eq. \eqref{eq:ccf-reflection-coefficients}. Substituting~\eqref{eq:ccf-reflection-coefficients-os2} into eq.~\eqref{eq:acf-equivalent-channel-simplified} {and considering that $| \mathbf{\bar{g}}^T \mathbf{\bar{h}} |^2 = \sum_{n = 1}^N \sum_{n' = 1}^N \bar{g}_n\bar{g}_{n'}^* \bar{h}_n\bar{h}_{n'}^*$} results in the \gls{acf} of the equivalent channel in eq. \eqref{eq:acf-equivalent-channel-os2} at the bottom of the page, with
\setcounter{equation}{12} 
\begin{equation}
    \label{eq:eta}
    \eta = \left| \frac{\mathbf{\bar{g}}^T\mathbf{\bar{h}}}{\norm{\mathbf{\bar{g}}}{2} \norm{\mathbf{\bar{h}}}{2}} \right|^2.
\end{equation}
Specifically, the equality $\mathrm{(a)}$ in eq.~\eqref{eq:acf-equivalent-channel-os2} results from the Rice factor in eq.~\eqref{eq:rice-factor}, \textit{i.e.}, from substituting the term $\norm{\mathbf{\bar{h}}}{2}^2 = \sum_{n = 1}^N \left|\bar{h}_n\right|$ by $N\kappa\sigma^2$. Also, from the triangle inequality, $\eta$ lies between $[0,1]$ and is a measure of \textit{orthogonality} between the \gls{los} components of the channels $\mathbf{g}[k]$ and $\mathbf{h}[k]$, depending on $N$ and the positions of the \gls{ris}, \gls{tx}, and \gls{rx} \cite{Albanese2022}.

By using the results for the \gls{ccf} and \gls{acf}, we derive the correlation coefficient between two channel samples delayed by $|\tau|$ samples, as shown in eq.~\eqref{eq:correlation-coefficient-os2} at the bottom of the page. One can notice that the temporal channel correlation depends on: $i$) the delay $|\tau|$ between the channel samples in discrete time, $ii$) the environmental parameters set $\mathcal{E}$, $iii$) the number of \gls{ris}' reflecting elements $N$, and $iv$) the parameter that determines the range of the phase shifts' distribution $\theta$. We make the following remarks about the obtained result.

\begin{rem}
    Regarding $\kappa$ and $\eta$, it is worth noting that, in the absence of a \gls{los} path component from the \gls{ris} to the receiver, \textit{i.e.}, {$\kappa \rightarrow -\infty$ dB}, or when the \gls{los} components are perfectly orthogonal, \textit{i.e.}, $\eta = 0$, the channel correlation is determined only by $\alpha$ and $\theta$. On the other hand, if  $\kappa > 0$ and $\eta > 0$, the correlation coefficient can also be altered by setting $N$.
    \label{remark:num-elements}
\end{rem}

\begin{rem}
    Regarding the distribution parameter $\theta$, we analyze how the temporal channel correlation behaves in the extreme values of its range, $[0,\pi]$. When $\theta = 0$, the resulting \textit{correlation coefficient} from \eqref{eq:correlation-coefficient-os2} simplifies to:
    \setcounter{equation}{14} 
    \begin{equation}
        \label{eq:correlation-coefficient-theta-zero}
        \left. \rho[\tau] \right|_{\theta = 0} = \frac{N\kappa\eta + \alpha^{|\tau|}}{N\kappa\eta + 1}.
    \end{equation}
    In this case, note that one can control the temporal correlation \textit{only} by selecting the number of \gls{ris}' elements $N$. On the other hand, when $\theta = \pi$, the correlation coefficient becomes:
    \begin{equation}
        \left. \rho[\tau] \right|_{\theta = \pi} = \delta[\tau].
    \end{equation}
    Now, observe that the channel samples are totally uncorrelated, corroborating with the findings of \cite{Besser2021}. Recall that the authors of \cite{Besser2021} used $\phi_n'[k] \uniformdist{0}{2\pi}$ to transform a slow-fading channel into a fast-fading one, improving reliability-related metrics. Hence, by tuning $\theta$ and given $N$, we can control the temporal correlation to values in the interval $[\left. \rho[\tau] \right|_{\theta = \pi}, \left. \rho[\tau] \right|_{\theta = 0}]$.
    \label{remark:theta}
\end{rem}

\subsection{Controlling the Temporal Correlation}\label{subsec:control}

Based on Remark \ref{remark:theta}, we present a method to design the {phase shifts distribution} parameter $\theta$ to obtain the desired channel correlation between samples separated from each other by a desired delay. This is based on the following: 

\begin{defn}
    (Project Requirements) The tuple of \textit{project requirements} is $p = (\tilde{\rho}, \tilde{\tau})$, where $0 \leq \tilde{\rho} \leq 1$ is the desired correlation coefficient of two-channel samples delayed by $\tilde{\tau} \in \integerset_+$ samples.
\end{defn}

\noindent\textit{Method.} From eq.~\eqref{eq:correlation-coefficient-os2} and for a constant $N$, the value of $\theta$ for obtaining a channel correlation coefficient $\rho[\pm \tilde{\tau}]=\tilde{\rho}$ is:
\begin{equation}
    \label{eq:theta-design-os2}
    \theta=\underline{\theta}(p) = \sinc^{-1} \left( \sqrt{\frac{(\kappa + 1) \tilde{\rho}}{\tilde{\rho} \kappa + (1 - \tilde{\rho}) N\kappa\eta + \alpha^{|\tilde{\tau}|}}} \right),
\end{equation}
where $\sinc^{-1}(\cdot)$ is the inverse function of $\sinc(\cdot)$ with codomain over the interval $[0, \pi]$.\footnote{In the domain $[0,\pi]$, the $\sinc(\cdot)$ function is partially invertible since it becomes bijective. In the absence of a closed-form expression for $\sinc^{-1}(\cdot)$, numerical methods can be used to calculate it with the required precision.}
By taking into account that $\theta \in [0,\pi]$, the argument in the \gls{rhs} of \eqref{eq:theta-design-os2} must lie in the interval $[0,1]$. Given this, we define the set of feasible project requirements as:
\begin{equation}
    \hspace{-1mm}\mathcal{P}_{\mathrm{feas.}}^{(\theta)} = \set{\left(\tilde{\rho}, \tilde{\tau}\right) \in \realset_+ \times \integerset_+ \left\vert\, 0 \leq \tilde{\rho} \leq \frac{N\kappa\eta + \alpha^{|\tilde{\tau}|}}{N\kappa\eta + 1} \right.}.
\end{equation}
From the above, one can note that the feasible channel correlation is upper-bounded by the environmental parameters $\mathcal{E}$ and the number of \gls{ris} reflecting elements $N$.

\begin{rem}
    \label{rem:N}
    In Remark \ref{remark:num-elements}, we have observed that one can also control the correlation coefficient by changing $N$. We now give an additional result showing the achievable channel correlation when opting for designing $N$.
    {We assume that the \gls{ris}, \gls{tx}, and \gls{rx} are positioned in a way that $\eta$ does not vary with $N$.\footnote{{We leave the analysis of the case where $\eta$ varies with $N$ for future works.}}}
    Considering a constant $\theta$, the value of $N$ for obtaining a correlation coefficient $\rho[\pm \tilde{\tau}]=\tilde{\rho}$ is:
    \begin{equation}
        \label{eq:N-design-os2}
        \underline{N}(p) = \floor{\frac{\left( 1 - \sinc^2(\theta) \right) (\kappa + 1) + \sinc^2(\theta) \left( \tilde{\rho} - \alpha^{|\tilde{\tau}|} \right)}{\sinc^2(\theta) (1 - \tilde{\rho}) \kappa\eta}},
    \end{equation}
    with $N=\underline{N}(p)$. Knowing that $N \geq 1$ and that the denominator of the argument at the \gls{rhs} of \eqref{eq:N-design-os2} must be nonzero, the set of feasible project requirements can be derived as:
    \begin{gather}
        \label{eq:feasible-project-requirements-N-os2}
        \mathcal{P}_{\mathrm{feas.}}^{(N)} = \set{\left. \left(\tilde{\rho}, \tilde{\tau}\right) \in \realset_+ \times \integerset_+ \right\vert \tilde{\rho}_{\min} \leq \tilde{\rho} < 1 }, \text{where}\\
        \hspace{-1mm}\tilde{\rho}_{\min} = \frac{\sinc^2(\theta) \left( \kappa\eta + \alpha^{|\tilde{\tau}|} \right) - \left( 1 - \sinc^2(\theta) \right) (\kappa + 1)}{\sinc^2(\theta) \left( \kappa\eta + 1 \right)}.
    \end{gather}
    From this, one can note that the achievable channel correlation by setting the number of \gls{ris}' elements $N$ is lower-bounded by the environmental parameters $\mathcal{E}$ and the parameter $\theta$.
\end{rem}

To get an overview of the condition in which Remark \ref{rem:N} is valid, Fig.~\ref{fig:eta-rx-position} depicts how $\eta$ changes with $N$ considering different \gls{rx} positions. The \gls{los} channel vectors are calculated with the model in \cite{Albanese2022}. Considering the \textit{right-handed Cartesian} coordinates system, the \gls{ris} is placed parallel to the $xy$-plane with center at coordinates $(0,0,5)$, while the \gls{tx} is at coordinates $(-10,0,0)$.
When the \gls{rx} position is symmetric to the \gls{tx} one w.r.t. the \gls{ris} center, $\eta = 1$ is constant. So Remark \ref{rem:N} is valid for position $(10,0,0)$. In position $(15,0,0)$, it may be valid for $N < 100$ due to the low variation of $\eta$ in this region. However, it does not apply for positions $(10,2,0)$ and $(10,10,0)$ due to the high-amplitude oscillations of $\eta$ with $N$.

\begin{figure}[b]
    \centering
    \includegraphics[width=\columnwidth]{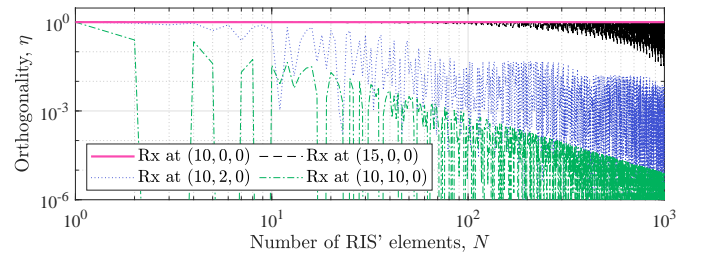}
    \vspace{-7mm}
    \caption{Orthogonality between the \gls{los} components as a function of $N$.}
    \label{fig:eta-rx-position}
\end{figure}

\section{Simulation Results}

In this section, we exemplify by numerical results how the method proposed in Subsection \ref{subsec:control} can be applied to obtain a given project requirement $p=(\tilde{\rho}, \tilde{\tau})$. The results and their respective simulation parameters are given in Figs.~\ref{fig:simulation-results-theta-N} and~\ref{fig:simulation-results-alpha-kappa}. {In the simulations, the coordinate system, the \gls{ris} and \gls{tx} positions, and the method to compute the \gls{los} channel components are the same as in Fig. \ref{fig:eta-rx-position}.}

\begin{figure*}[t]
    \subfloat[
    {$\alpha = 1 - 1.12 \times 10^{-4}$, $\kappa = 6$ dB $N = 100$, \gls{rx} at $(10,10,0)$}]{%
        \includegraphics[width=\columnwidth]{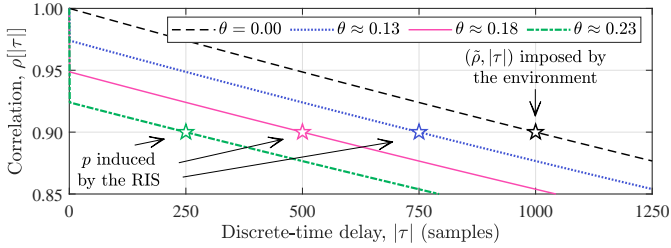}%
        \label{fig:channel-correlation-theta}%
    }
    \hfill
    \subfloat[{$\alpha = 0.992$, $\kappa = -6$ dB, $\theta = 0$, \gls{rx} at $(15,0,0)$}]{%
        \includegraphics[width=\columnwidth]{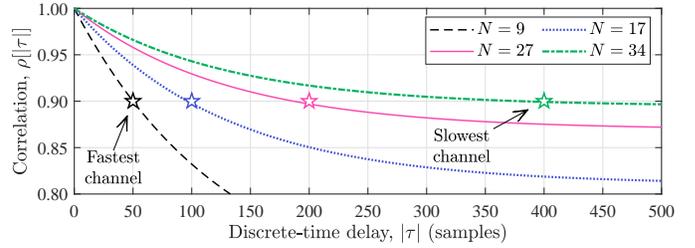}%
        \label{fig:channel-correlation-N}%
    }
    \caption{Channel temporal correlation. The markers indicate the points where the correlation should reach $0.9$ according to the project requirements $p$.}
    \label{fig:simulation-results-theta-N}
\end{figure*}

Fig.~\ref{fig:channel-correlation-theta} depicts the correlation achieved by different $\theta$ with a fixed $N = 100$. {From \eqref{eq:ris-phase-shifts}, recall that $\theta = 0$ implies static \gls{ris} phase shifts equal to $\pi$. Under this condition, $\alpha$ in this result is calculated by eq. \eqref{eq:correlation-coefficient-theta-zero} to obtain $\rho[1000] |_{\theta = 0} = 0.9$.} For the cases where $\theta > 0$, {the phase shifts distribution parameter} $\theta$ is calculated by eq. \eqref{eq:theta-design-os2} to obtain a correlation of $\tilde{\rho} = 0.9$ at the time delays $\tilde{\tau} \in \set{250, 500, 750}$. This result shows that the proposed method can change the temporal channel statistics imposed by the environment. In the sequel, Fig.~\ref{fig:channel-correlation-N} depicts the correlation obtained by different $N$. {In this result, $\alpha$ is calculated using the aforementioned method, but now to obtain $\rho[50] |_{\theta = 0,N=9} = 0.9$.} The values for {$N > 9$} are calculated by eq.~\eqref{eq:N-design-os2} to yield a correlation  of $\tilde{\rho} = 0.9$ at the time delays $\tilde{\tau} \in \set{100, 200, 400}$. Such a result reveals that a fast channel, \textit{i.e.}, a channel with a fast decay correlation, can be slowed down by increasing the number of elements of the \gls{ris}.

\begin{figure}[b]
    \subfloat[{$\kappa = 6$ dB, \gls{rx} at $(10,10,0)$}]{%
        \includegraphics[width=.5\columnwidth]{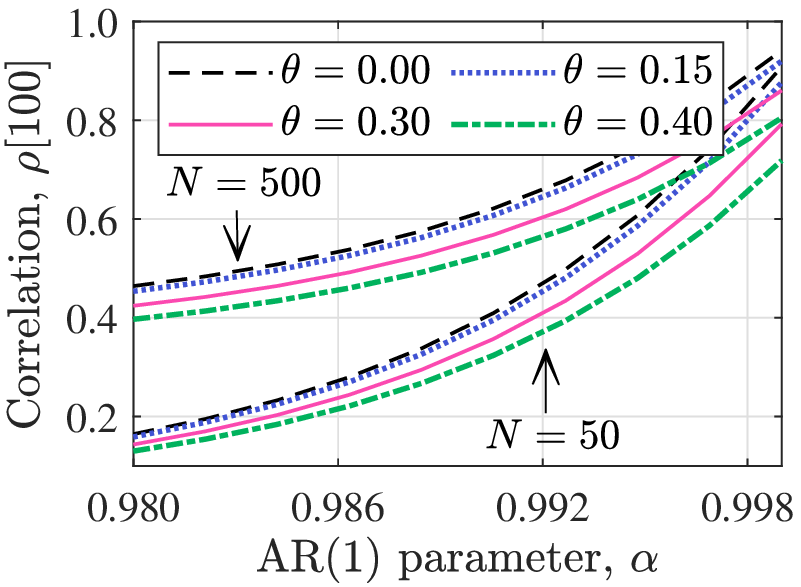}%
        \label{fig:channel-correlation-alpha2}%
    }
    \subfloat[{$N = 100$, $\theta = 0$, Rx at $(10,2,0)$}]{%
        \includegraphics[width=.5\columnwidth]{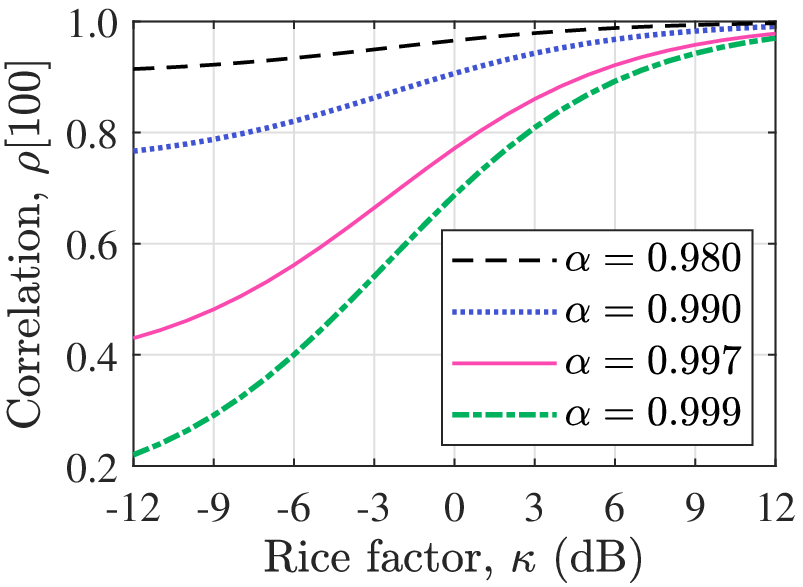}%
        \label{fig:channel-correlation-kappa}%
    }
    \caption{{Channel temporal correlation as a function of (a) $\alpha$ and (b) $\kappa$.}}
    \label{fig:simulation-results-alpha-kappa}
\end{figure}

Fig.~\ref{fig:channel-correlation-alpha2} shows {the impact of} $\theta$ and $N$ in modifying the correlation at $|\tau| = 100$ as a function of $\alpha$, representing different environmental conditions. It demonstrates that the correlation changes quickly with little change in $\alpha$, justifying the use of the \gls{ar1} model to represent both slow- and fast-fading channels. Then, Fig.~\ref{fig:channel-correlation-kappa} depicts the correlation at $|\tau| = 100$ as a function of $\kappa$ and under different $\alpha$.
This result reveals that high $\kappa$ yields correlation values close to $1$ due to the significant increase on the power of the deterministic part of the channel $\mathbf{\bar{h}}$ relatively to the stochastic one $\{\check{h}_n\}_{n = 1}^N$. In other words, as expected, the higher the $\kappa$ parameter, the lower the impact of the \gls{ris} in controlling the environment. It is worth mentioning that low $\kappa$ is typical in scenarios with partially blocked  \gls{los} and/or environments with rich scattering.

\section{Towards a Flexible Block-Fading Model}

In this section, we extend the classical block-fading model \cite{Bjornson2017,Goldsmith2005} to account for the channel correlation control based on the randomized framework proposed in Section \ref{sec:rand-framework}. Define a \textit{coherence block} as a resource block consisting of a number of subcarriers and time samples where the equivalent channel response $z$ can be approximated as constant and flat-fading. Specifically, each coherence block has $\Delta_c=B_c T_c$ complex-valued samples, where $T_c > 0$ is the channel coherence time. Moreover, the channel response (power gain) $z$ of this discrete-time channel follows a given distribution $z \sim {f_Z}$. For example, for Rayleigh fading channels, ${f_Z}$ is exponential. The \gls{ris}-enabled control discussed in Section \ref{sec:rand-framework} can be used to create coherent blocks with different lengths, where the coherence time relates to the {discrete-time} interval $\tilde{\tau}$. Recall that the channel coherence time is defined as the range of time span values over which the channel \gls{acf} is approximately nonzero \cite{Goldsmith2005}. Therefore, using the \gls{ris} to shape the \gls{acf} of the equivalent channel is a path to control $T_c$ and, consequently, changing $\Delta_c$. This generation of coherence blocks with a flexible number of samples can be done by setting $\theta$ and $N$ {to obtain} a project requirement $p$ as described respectively by eqs.~\eqref{eq:theta-design-os2} and~\eqref{eq:N-design-os2}.
Particularly, this flexible block-fading model can improve how the resources are leveraged, enabling the on-demand creation of blocks according to the availability of services with different performance requirements over time.

\section{Conclusion}\label{sec:conclusion}

In this paper, we have studied how an \gls{ris} can change the temporal statistics of the wireless
propagation channel by analyzing the correlation among channel samples using the introduced discrete-time-varying channel model. Then, we proposed a randomized framework to control the relative channel coherence time by setting the number of \gls{ris}' elements and {designing} the distribution of their reflection coefficients, whose effectiveness is corroborated by simulation results. {Our results demonstrate} the possibility of redefining the resource allocation problem as we know it today by creating a flexible block-fading model based on the proposed framework.

\appendices
\section{Proof of the ACF of the Equivalent Channel}\label{apx:acf-equivalent-channel}

Recalling that $h_{\mathrm{D}}[k]$, $\mathbf{g}[k]$, $\mathbf{h}[k]$, and $\boldsymbol{\Psi}[k]$ are mutually independent, the \gls{acf} of eq. \eqref{eq:equivalent-channel} is given by the sum:
\begin{gather}
    \label{eq:Rhappendix}
    R_{h_{\mathrm{eq}}h_{\mathrm{eq}}}[\tau] = \expval{h_{\mathrm{eq}}[k]h_{\mathrm{eq}}^*[k-\tau]} = S_1 + S_2 + S_3,\\
   \text{where } S_1 = \expval{h_{\mathrm{D}}[k] h_{\mathrm{D}}^*[k - \tau]},\\
    \begin{align}
        S_2 &= \mathbb{E} \Big\{ \left(\mathbf{g}^T[k] \boldsymbol{\Psi}[k] \mathbf{h}[k]\right) \times\\
        \nonumber
        &\times \left(\mathbf{g}^T[k-\tau] \boldsymbol{\Psi}[k-\tau] \mathbf{h}[k-\tau]\right)^* \Big\}, \text{ and}
    \end{align}\\
    \begin{align}
        S_3 &= \expval{h_{\mathrm{D}}[k] \left(\mathbf{g}^T[k-\tau] \boldsymbol{\Psi}[k-\tau] \mathbf{h}[k-\tau]\right)^*} +\\
        \nonumber
        &+ \mathbb{E} \Big\{ h_{\mathrm{D}}^*[k-\tau] \left(\mathbf{g}^T[k] \boldsymbol{\Psi}[k] \mathbf{h}[k]\right) \Big\}.
    \end{align}
\end{gather}
Now we evaluate each term independently.
Initially, due to the model in \eqref{eq:rician-model} adopted for $h_{\mathrm{D}}[k]$, $S_1$ can be rewritten as:
\begin{align}
    \nonumber
    S_1 &= \expval{\left(\bar{h}_{\mathrm{D}} + \check{h}_{\mathrm{D}}[k]\right) \left(\bar{h}_{\mathrm{D}} + \check{h}_{\mathrm{D}}[k-\tau]\right)^*}\\
    &= \left|\bar{h}_{\mathrm{D}}\right|^2 + \expval{\check{h}_{\mathrm{D}}[k] \check{h}_{\mathrm{D}}^*[k-\tau]} +\\
    \nonumber
    &+ \expval{\bar{h}_{\mathrm{D}} \check{h}_{\mathrm{D}}^*[k-\tau]} + \expval{\bar{h}_{\mathrm{D}}^* \check{h}_{\mathrm{D}}[k]}.
\end{align}
Since $\bar{h}_{\mathrm{D}}$ is deterministic and $\check{h}_{\mathrm{D}}[k]$ is an \gls{ar1} random process as per \eqref{eq:multipath-component}, the expectations are given respectively by:
\begin{equation}
\begin{cases}
    \expval{\check{h}_{\mathrm{D}}[k] \check{h}_{\mathrm{D}}^*[k-\tau]} = \alpha_{h_{\mathrm{D}}}^{|\tau|} \sigma_{h_{\mathrm{D}}}^2, \\
    \expval{\bar{h}_{\mathrm{D}} \check{h}_{\mathrm{D}}^*[k-\tau]} = \expval{\bar{h}_{\mathrm{D}}^* \check{h}_{\mathrm{D}}[k]} = 0.
\end{cases}    
\end{equation}
So, $S_1$ can be rewritten as:
\begin{equation}
    \label{eq:s1-final}
    S_1 = \left|\bar{h}_{\mathrm{D}}\right|^2 + \alpha_{h_{\mathrm{D}}}^{|\tau|} \sigma_{h_{\mathrm{D}}}^2.
\end{equation}

Expanding the multiplications in $S_2$ results in
\begin{gather}
    S_2 = \textstyle\sum_{n = 1}^N \sum_{n' = 1}^N P_1 P_2 P_3,\\
    \text{where }
    \begin{cases}
        P_1 = \expval{g_n[k] g_{n'}[k-\tau]},\\
        P_2 = \expval{h_n[k] h_{n'}[k-\tau]}, \\
        P_3 = \expval{\psi_n[k] \psi_{n'}[k-\tau]} = R_{\psi_n\psi_{n'}}[\tau],
    \end{cases}    
\end{gather}
with $R_{\psi_n\psi_{n'}}$ defined in eq.~\eqref{eq:ccf-reflection-coefficients}.
From the model in eq.~\eqref{eq:rician-model} adopted for $\mathbf{g}[k]$ and $\mathbf{h}[k]$, and recalling that $\{g_n[k]\}_{n = 1}^N$ and $\{h_n[k]\}_{n = 1}^N$ are mutually independent, $P_1$ and $P_2$ result
\begin{align}
    P_1 &=
    \begin{cases}
        \left|\bar{g}_n\right|^2 + \alpha_{g_n}^{|\tau|} \sigma_{g_n}^2, & \textrm{if} \; n = n'\\
        \bar{g}_n\bar{g}_{n'}^*, & \textrm{otherwise}
    \end{cases}, \\
    P_2 &=
    \begin{cases}
        \left|\bar{h}_n\right|^2 + \alpha_{h_n}^{|\tau|} \sigma_{h_n}^2, & \textrm{if} \; n = n'\\
        \bar{h}_n\bar{h}_{n'}^*, & \textrm{otherwise}
    \end{cases}.
\end{align}
While the results for $n = n'$ are based on the derivation of $S_1$, the results for $n \neq n'$ come from the mean of the time-variant channel coefficient in eq. \eqref{eq:rician-model-moments}. Therefore, given the results for $P_1$, $P_2$, and $P_3$, we can rewrite $S_2$ as:
\begin{align}
    \label{eq:s2-final}
    S_2 &=  \textstyle\sum_{n = 1}^N \sum_{n' = 1, n' \neq n}^N \bar{g}_n\bar{g}_{n'}^* \bar{h}_n\bar{h}_{n'}^* R_{\psi_n\psi_{n'}}[\tau] \; +\\
    \nonumber
    &+ \textstyle\sum_{n = 1}^N \left(\left|\bar{g}_n\right|^2 + \alpha_{g_n}^{|\tau|} \sigma_{g_n}^2\right) \left(\left|\bar{h}_n\right|^2 + \alpha_{h_n}^{|\tau|} \sigma_{h_n}^2\right) R_{\psi_n\psi_n}[\tau].
\end{align}

In $S_3$, the expectation of each multiplication can be rewritten as the multiplication of the expected value of each term:
\begin{equation}
    S_3 = \bar{h}_{\mathrm{D}} \mathbf{\bar{g}}^H \expval{\boldsymbol{\Psi}^*[k-\tau]} \mathbf{\bar{h}}^* + \bar{h}_{\mathrm{D}}^* \mathbf{\bar{g}}^T \expval{\boldsymbol{\Psi}[k]} \mathbf{\bar{h}}.
\end{equation}
Expanding the results of the multiplications in the summations:
\begin{equation}
    S_3 = \bar{h}_{\mathrm{D}} \sum_{n = 1}^N \bar{g}_n^* \bar{h}_n^* \expval{\psi_n^*[k-\tau]} + \bar{h}_{\mathrm{D}}^* \sum_{n = 1}^N \bar{g}_n \bar{h}_n \expval{\psi_n[k]}.
\end{equation}
Considering that $\{\psi_n[k]\}_{n = 1}^N$ are \gls{wss} random processes, $S_3$ can be rewritten as:
\begin{equation}
    \label{eq:s3-final}
    \textstyle
    S_3 = 2 \realpart{\bar{h}_{\mathrm{D}}^* \sum_{n = 1}^N \bar{g}_n \bar{h}_n \expval{\psi_n[k]}}.
\end{equation}

Finally, substituting eqs.~\eqref{eq:s1-final},~\eqref{eq:s2-final}, and~\eqref{eq:s3-final} into eq.~\eqref{eq:Rhappendix}, we obtain eq.~\eqref{eq:acf-equivalent-channel}, completing the proof. \hfill $\qed$

\bibliographystyle{IEEEtran}
\bibliography{references}

\end{document}